\documentclass[11pt]{amsart}

\usepackage[utf8]{inputenc}
\usepackage[T1]{fontenc}
\usepackage{lmodern}
\usepackage[a4paper,margin=1in]{geometry}
\usepackage{microtype}
\usepackage{amsmath,amssymb,amsthm,mathtools}
\usepackage{mathrsfs}
\usepackage{enumitem}
\usepackage{booktabs}
\usepackage{longtable}
\usepackage{array}
\usepackage{graphicx}
\usepackage{needspace}
\usepackage[hidelinks]{hyperref}
\usepackage[nameinlink,noabbrev]{cleveref}
\usepackage{tikz}
\usetikzlibrary{matrix,positioning,calc,arrows.meta,fit,shapes.geometric,backgrounds}
\usepackage{caption}

\allowdisplaybreaks[2]
\numberwithin{equation}{section}

\hypersetup{
  pdftitle={Closure Atlases and Local-to-Global Obstructions in Finite Closure Systems},
  pdfauthor={Jaehwan Kim}
}

\theoremstyle{plain}
\newtheorem{theorem}{Theorem}[section]
\newtheorem{lemma}[theorem]{Lemma}
\newtheorem{proposition}[theorem]{Proposition}
\newtheorem{corollary}[theorem]{Corollary}

\theoremstyle{definition}
\newtheorem{definition}[theorem]{Definition}
\newtheorem{example}[theorem]{Example}
\newtheorem{counterexample}[theorem]{Counterexample}

\theoremstyle{remark}
\newtheorem{remark}[theorem]{Remark}

\crefname{definition}{Definition}{Definitions}
\Crefname{definition}{Definition}{Definitions}
\crefname{example}{Example}{Examples}
\Crefname{example}{Example}{Examples}
\crefname{counterexample}{Counterexample}{Counterexamples}
\Crefname{counterexample}{Counterexample}{Counterexamples}
\crefname{construction}{Construction}{Constructions}
\Crefname{construction}{Construction}{Constructions}
\crefname{theorem}{Theorem}{Theorems}
\Crefname{theorem}{Theorem}{Theorems}
\crefname{proposition}{Proposition}{Propositions}
\Crefname{proposition}{Proposition}{Propositions}
\crefname{lemma}{Lemma}{Lemmas}
\Crefname{lemma}{Lemma}{Lemmas}
\crefname{corollary}{Corollary}{Corollaries}
\Crefname{corollary}{Corollary}{Corollaries}
\crefname{remark}{Remark}{Remarks}
\Crefname{remark}{Remark}{Remarks}
\crefname{figure}{Figure}{Figures}
\Crefname{figure}{Figure}{Figures}
\crefname{table}{Table}{Tables}
\Crefname{table}{Table}{Tables}

\newcommand{\Pow}{\mathcal P}
\newcommand{\Cl}{\operatorname{cl}}
\newcommand{\Fix}{\operatorname{Fix}}
\newcommand{\C}{\mathcal C}
\newcommand{\X}{\mathfrak X}

\newcommand{\F}{\mathfrak F}
\newcommand{\A}{\mathfrak A}
\newcommand{\Sig}{\Sigma}
\newcommand{\Obs}{\operatorname{Obs}}
\newcommand{\CA}{C_{\mathfrak A}}
\newcommand{\GA}{G_{\mathfrak A}}
\newcommand{\Tcup}{T^{\cup}}
\newcommand{\restr}[1]{\upharpoonright_{#1}}

\title[Closure atlases and local-to-global obstructions]{Closure Atlases and Local-to-Global Obstructions in Finite Closure Systems}
\author{Jaehwan Kim}
\address{Hankuk Academy of Foreign Studies, Yongin-si, Gyeonggi-do 17035, Republic of Korea}
\email{020080@hafs.hs.kr}
\date{June 19, 2026}

\begin{document}

\begin{abstract}
This paper studies finite closure operators on overlapping finite universes and gives an exact local-to-global obstruction criterion for conservative globalization.  Given a finite family of local closure systems, its atlas-generated closure is obtained by repeatedly applying the local closure operators to the parts visible in each chart.  This closure is the least global closure operator extending all chart closures.  A chart-visible obstruction is a consequence produced by this global propagation that lies inside a chart but is not validated by that chart's own closure operator.  The main theorem proves that a finite closure atlas has a global conservative realization exactly when no such obstruction occurs; in that case the atlas-generated closure itself is the conservative realization.  The obstruction condition is finite and directly computable.

The paper also records the indexed representation layer motivating the terminology.  For a finite closure system, an indexed truth space selects closed theories as contexts and represents each element by the region of selected closed theories containing it.  Closure consequence is always sound for region inclusion, and the full indexed space of all closed theories recovers the original closure consequence exactly; reduced indexed spaces can therefore create spurious region consequences by deleting separating closed theories.  A formal opposite gives a four-region membership decomposition--only one, only the other, both, and neither--unless additional separation assumptions are imposed.  Finally, overlap-compatible local closed theories glue by canonical union under the atlas-generated closure.  The framework is finite, structural, and closure-theoretic; the logical terminology is used only as an interpretation of the underlying closure data.
\end{abstract}

\maketitle

\setcounter{tocdepth}{1}
\tableofcontents

\section{Introduction}
\label{sec:introduction}

This paper studies finite closure operators on finite sets and their local-to-global behavior under overlaps.  The guiding problem is structural: given several local closure systems on overlapping finite universes, when can they be realized as exact restrictions of one global closure operator, and what obstruction prevents such a realization?  The indexed truth-space language used below is one interpretation of the same closure data.  The technical content is closure-theoretic: least generated closures, conservative restriction, finite obstruction detection, and gluing of compatible closed sets.

The starting point is a finite closure system
\[
  \C=(\Sig,\Cl),
\]
where \(\Sig\) is a finite universe, often read as a set of sentences, and \(\Cl:\Pow(\Sig)\to\Pow(\Sig)\) is an extensive, monotone, idempotent operator.  A closed set \(T=\Cl(T)\) is called a closed theory when the elements of \(\Sig\) are read logically.  For \(S\subseteq\Sig\) and \(\varphi\in\Sig\), the notation \(S\vdash_{\Cl}\varphi\) means precisely \(\varphi\in\Cl(S)\).  Thus the logical terminology is not an additional semantics; it is notation for a finite closure operator.

The central construction is a finite closure atlas
\[
  \A=\{(\Sig_\alpha,\Cl_\alpha)\}_{\alpha\in A}
\]
whose local universes may overlap.  Its global universe is \(\Sig_{\A}=\bigcup_{\alpha\in A}\Sig_\alpha\).  Local propagation is encoded by
\[
  G_{\A}(X)=X\cup \bigcup_{\alpha\in A}
       \Cl_\alpha(X\cap \Sig_\alpha),
  \qquad X\subseteq\Sig_{\A},
\]
and iterating this finite monotone operator to a fixed point yields the atlas-generated closure \(C_{\A}\).  This is the least global closure operator extending all local chart closures.  In particular, any global closure operator that preserves the local implications must contain every consequence generated by \(C_{\A}\).

The main theorem gives an exact obstruction criterion for conservative globalization.  A global conservative realization is a global closure operator \(D\) such that every chart is recovered exactly by restriction:
\[
  D(S)\cap\Sig_\alpha=\Cl_\alpha(S)
  \qquad(\alpha\in A,\ S\subseteq\Sig_\alpha).
\]
A chart-visible obstruction is a triple \((\alpha,S,\varphi)\) such that \(S\subseteq\Sig_\alpha\), \(\varphi\in\Sig_\alpha\), and
\[
  \varphi\in C_{\A}(S)
  \qquad\text{but}\qquad
  \varphi\notin \Cl_\alpha(S).
\]
Thus the least global propagation has forced, inside chart \(\alpha\), a consequence not sanctioned by that chart itself.  The theorem proves
\[
  \Obs(\A)=\varnothing
  \quad\Longleftrightarrow\quad
  \A\text{ has a global conservative realization},
\]
and in the obstruction-free case \(C_{\A}\) itself is that realization.  Since all universes are finite, the criterion also gives a direct finite algorithm: iterate \(G_{\A}\) on every chart-local premise set and compare the result with the corresponding local closure.

The single-system material supplies the representation and comparison language used before the atlas theorem.  An indexed truth space over \(\C\) selects a finite family of closed theories,
\[
  \X=(\Sig,\Cl,I,\tau),
  \qquad \tau:I\to \Fix(\Cl),
\]
and represents each element \(\varphi\in\Sig\) by the region
\[
  V_{\X}(\varphi)=\{i\in I:\varphi\in\tau(i)\}.
\]
Closure consequence always implies truth-region inclusion.  Conversely, when all closed theories in \(\Fix(\Cl)\) are used as indices, truth-region inclusion recovers closure consequence exactly.  This identifies precisely what is lost when one replaces a closure system by a reduced family of selected closed theories: omitted closed theories may be the separating witnesses that distinguish nonconsequence from consequence.

A second auxiliary layer is decision bookkeeping.  For an element \(a\) and a designated formal opposite \(\bar a\), no logical behavior of \(\bar a\) is assumed from the notation alone.  The selected index set decomposes into four regions: indices containing \(a\) but not \(\bar a\), indices containing \(\bar a\) but not \(a\), indices containing both, and indices containing neither.  The familiar positive-negative-undecided trichotomy appears only after imposing an additional opposite-separation condition.  This layer is not used in the proof of the local-to-global obstruction theorem; it only separates finite membership bookkeeping from any implicit semantics of negation.

Conservative extension supplies the bridge between the single-system and atlas settings.  If \(\C_Y=(\Sig_Y,\Cl_Y)\) extends \(\C_X=(\Sig_X,\Cl_X)\), conservativity means that no new old-universe consequences are introduced:
\[
  \Cl_Y(S)\cap \Sig_X=\Cl_X(S)
  \qquad(S\subseteq\Sig_X).
\]
For full indexed truth spaces, this is equivalent to an exact pullback identity for old truth regions along reducts of closed theories.  The atlas theorem imposes the same conservative-restriction requirement simultaneously for several overlapping charts.

The final technical result is a canonical gluing statement.  After fixing the atlas-generated closure, any family of locally closed theories that agrees exactly on overlaps has a union that restricts back to the original local theories and is closed under \(C_{\A}\).  This result concerns compatible local closed theories; it is separate from the obstruction theorem, which concerns conservative realization of local consequence relations.

The formal results of the paper do not depend on any stronger philosophical thesis.  They are finite statements about closure operators, closed sets, local propagation, and conservative restriction.  The terms ``truth region,'' ``closed theory,'' and ``formal opposite'' are used to motivate one intended interpretation, but the definitions and proofs stand as order-theoretic closure constructions.

The paper is organized as follows.  \Cref{sec:related-work-scope} places the manuscript in relation to standard closure-theoretic, order-theoretic, model-theoretic, and local-to-global terminology.  \Cref{sec:finite-closure-systems} introduces finite closure systems, indexed truth spaces, truth regions, and the full truth-region representation theorem.  \Cref{sec:decision-decompositions} records the auxiliary four-region decomposition for formal opposites.  \Cref{sec:conservative-closure-extensions} proves the conservative-extension pullback theorem for full truth spaces.  \Cref{sec:closure-atlases} defines finite closure atlases and the atlas-generated closure.  \Cref{sec:obstructions} proves the local-to-global obstruction theorem and gives the finite obstruction algorithm.  \Cref{sec:closed-theory-gluing} proves canonical gluing for compatible local closed theories.  \Cref{sec:conclusion} summarizes the scope of the framework and records conservative completions as future work.

\section{Related Work and Scope}
\label{sec:related-work-scope}

The primary technical background of the paper is the order-theoretic and lattice-theoretic study of closure operators.  A closure system on a finite set may be represented equivalently by an extensive, monotone, idempotent operator, by its lattice of closed sets, or by suitable finite implication data.  This places the formalism near classical lattice theory and closure systems \cite{Birkhoff1967,DaveyPriestley2002}, finite closure systems and implicational systems \cite{GuiguesDuquenne1986,CaspardMonjardet2003,BertetDemkoViaudGuerin2018}, and formal concept analysis, where finite closure operators and closed sets are central \cite{GanterWille1999}.  The present manuscript does not use concept lattices, formal contexts, or Galois connections as technical primitives; the comparison is meant only to locate the closure-theoretic background.

The terminology of consequence is Tarskian in spirit: a consequence relation is represented extensionally by a closure operator \cite{Tarski1956,Wojcicki1988}.  In this paper that logical reading is secondary to the closure-theoretic construction.  The atlas step operator is a finite forward propagation procedure for local implications, formally adjacent to Horn-style implications and dependency closure in database theory \cite{Fagin1982,Maier1983}.  The paper does not, however, study implication bases, database dependencies, or algorithmic minimization problems as objects in their own right.  Its use of propagation is purely closure-theoretic: local closures on overlapping finite universes are iterated until a least global closure is reached.

The terms ``conservative extension'' and ``reduct'' are used in a closure-theoretic form analogous to their standard use in logic and model theory: an expansion should not create new old-universe consequences.  The model-theoretic background can be found in standard references such as Hodges \cite{Hodges1993}, but no compactness, interpolation, elementary embeddings, or model-theoretic preservation theorem is used here.  Similarly, the terms ``atlas,'' ``local-to-global,'' and ``gluing'' are used by analogy with familiar local-to-global language in geometry and sheaf theory \cite{MacLaneMoerdijk1992}.  No topology, sheaf condition, or topos-theoretic machinery is assumed.

The paper is also distinct from abstract algebraic logic.  It does not assign algebraic semantics to a given deductive system, and it does not study algebraizability in the sense of \cite{BlokPigozzi1989,FontJansana1996}.  Its objects are finite closure systems, selected families of closed theories, and finite atlases of local closure systems.  The technical contribution is correspondingly narrower: a truth-region representation result for full indexed spaces, a conservative-extension pullback theorem, a finite local-to-global obstruction criterion, and a gluing theorem for compatible local closed theories.  The organization of the paper is summarized in \Cref{tab:paper-roadmap}.

\begin{table}[htbp]
\centering
\renewcommand{\arraystretch}{1.15}
\begin{tabular}{@{}p{0.27\textwidth}p{0.32\textwidth}p{0.31\textwidth}@{}}
\toprule
Layer & Main object & Output \\
\midrule
Indexed truth spaces
& \(V_{\mathfrak X}(\varphi)\)
& region consequence and full representation \\

Decision bookkeeping
& \(P_a,N_a,B_a,U_a\)
& four-region bookkeeping without implicit negation semantics \\

Conservative extension
& reduct map \(\rho\)
& exact pullback of old truth regions \\

Closure atlases
& \(G_{\mathfrak A}\) and \(C_{\mathfrak A}\)
& least global closure generated by local propagation \\

Obstruction theory
& \(\operatorname{Obs}(\mathfrak A)\)
& criterion for conservative globalization \\

Closed-theory gluing
& \(T^\cup\)
& canonical union of compatible local theories \\
\bottomrule
\end{tabular}
\caption{The logical progression of the paper.}
\label{tab:paper-roadmap}
\end{table}

\section{Finite Closure Systems and Indexed Truth Spaces}
\label{sec:finite-closure-systems}

We begin with the single-system case.  The finiteness assumption is not needed for every definition in this section, but it will be used later when the obstruction theorem is converted into an explicit finite algorithm.  The section has two roles: first, to fix the closure consequence relation of a finite closure operator; second, to compare that relation with the weaker information visible through a chosen indexed family of closed sets, called closed theories in the intended logical reading.  This distinction is the base case for every later comparison between local and global closure data.

\begin{definition}[Finite closure system, consequence, and closed theories]
A \emph{finite closure system} is a pair
\[
  \C=(\Sig,\Cl),
\]
where \(\Sig\) is a finite set of sentences and
\[
  \Cl:\Pow(\Sig)\to\Pow(\Sig)
\]
is extensive, monotone, and idempotent; that is, for all \(S,U\subseteq\Sig\),
\begin{enumerate}[label=(\roman*)]
  \item \(S\subseteq\Cl(S)\),
  \item if \(S\subseteq U\), then \(\Cl(S)\subseteq\Cl(U)\),
  \item \(\Cl(\Cl(S))=\Cl(S)\).
\end{enumerate}
For \(S\subseteq\Sig\) and \(\varphi\in\Sig\), write
\[
  S\vdash_{\C}\varphi
  \quad\Longleftrightarrow\quad
  \varphi\in\Cl(S).
\]
Equivalently, one may write \(S\vdash_{\Cl}\varphi\).  When the closure system is clear from context, we simply write \(S\vdash\varphi\).

A subset \(T\subseteq\Sig\) is a \emph{closed theory} of \(\C\) if \(\Cl(T)=T\).  Write
\[
  \Fix(\Cl)=\{T\subseteq\Sig: \Cl(T)=T\}
\]
for the set of all closed theories of \(\C\).
\end{definition}

\begin{remark}[Raw sets and closed theories]
A subset \(S\subseteq\Sig\) is an arbitrary set of sentences; it need not be closed. A closed theory is a stabilized object \(T=\Cl(T)\). Thus \(S\vdash\varphi\) does not mean that \(\varphi\in S\). It means that \(\varphi\) belongs to the closed theory generated by \(S\).
\end{remark}

\begin{definition}[Finite indexed truth space and truth regions]
Let \(\C=(\Sig,\Cl)\) be a finite closure system. A \emph{finite indexed truth space} over \(\C\) is a tuple
\[
  \X=(\Sig,\Cl,I,\tau),
\]
where \(I\) is a finite index set and
\[
  \tau:I\to\Fix(\Cl)
\]
assigns a closed theory to each index.  The elements of \(I\) are called \emph{indices} or, when helpful, \emph{contexts}. In the main examples and representation theorems below, the relevant index sets are nonempty.

For \(\varphi\in\Sig\), the \emph{truth region} of \(\varphi\) in \(\X\) is
\[
  V_{\X}(\varphi)=\{i\in I:\varphi\in\tau(i)\}.
\]
For \(S\subseteq\Sig\), define
\[
  V_{\X}(S)=\{i\in I:S\subseteq\tau(i)\}
  =\bigcap_{\psi\in S}V_{\X}(\psi),
\]
with \(V_{\X}(\varnothing)=I\).  When \(\X\) is clear from context, we may write \(V(\varphi)\) and \(V(S)\).
\end{definition}

\begin{definition}[Region consequence]
For \(S\subseteq\Sig\) and \(\varphi\in\Sig\), write
\[
  S\models_{\X}\varphi
\]
if
\[
  V_{\X}(S)\subseteq V_{\X}(\varphi).
\]
Thus every selected index whose assigned closed theory contains all sentences in \(S\) also contains \(\varphi\). This relation depends on the chosen indexed truth space \(\X\). It should not be confused with closure consequence \(S\vdash_{\C}\varphi\).
\end{definition}

\begin{lemma}[Truth-region soundness]
\label{lem:truth-region-soundness}
Let \(\X=(\Sig,\Cl,I,\tau)\) be a finite indexed truth space over \(\C=(\Sig,\Cl)\). For all \(S\subseteq\Sig\) and \(\varphi\in\Sig\),
\[
  S\vdash_{\C}\varphi
  \quad\Longrightarrow\quad
  V_{\X}(S)\subseteq V_{\X}(\varphi).
\]
Equivalently, \(S\vdash_{\C}\varphi\) implies \(S\models_{\X}\varphi\).
\end{lemma}

\begin{proof}
Assume \(S\vdash_{\C}\varphi\), so \(\varphi\in\Cl(S)\). Let \(i\in V_{\X}(S)\). Then \(S\subseteq\tau(i)\). By monotonicity of \(\Cl\),
\[
  \Cl(S)\subseteq\Cl(\tau(i)).
\]
Since \(\tau(i)\) is a closed theory, \(\Cl(\tau(i))=\tau(i)\). Hence \(\Cl(S)\subseteq\tau(i)\). Since \(\varphi\in\Cl(S)\), it follows that \(\varphi\in\tau(i)\), and therefore \(i\in V_{\X}(\varphi)\). This proves the inclusion.
\end{proof}

\begin{remark}[Soundness need not be completeness]
Lemma~\ref{lem:truth-region-soundness} says that closure consequence is always sound for truth-region inclusion. The converse need not hold for an arbitrary indexed truth space: the chosen index set may omit closed theories that distinguish \(S\) from \(\varphi\). The exact converse is recovered when all closed theories are used as indices.
\end{remark}

\begin{definition}[Full indexed truth space]
Let \(\C=(\Sig,\Cl)\) be a finite closure system. The \emph{full indexed truth space} of \(\C\) is
\[
  \F_{\C}=(\Sig,\Cl,\Fix(\Cl),\mathrm{id}),
\]
where the index set is \(\Fix(\Cl)\) and the assignment map is the identity, \(\mathrm{id}(T)=T\). Thus every closed theory is used as an index. Since \(\Sig\) is finite, \(\Fix(\Cl)\) is finite; it is nonempty because \(\Cl(\varnothing)\) is closed.
\end{definition}

\begin{theorem}[Full truth-region representation]
\label{thm:full-truth-region-representation}
Let \(\C=(\Sig,\Cl)\) be a finite closure system, and let \(\F_{\C}\) be its full indexed truth space. For every \(S\subseteq\Sig\) and every \(\varphi\in\Sig\),
\[
  V_{\F_{\C}}(S)\subseteq V_{\F_{\C}}(\varphi)
  \quad\Longleftrightarrow\quad
  S\vdash_{\C}\varphi.
\]
Equivalently,
\[
  S\models_{\F_{\C}}\varphi
  \quad\Longleftrightarrow\quad
  S\vdash_{\C}\varphi.
\]
More generally, for \(S,U\subseteq\Sig\),
\[
  V_{\F_{\C}}(S)\subseteq V_{\F_{\C}}(U)
  \quad\Longleftrightarrow\quad
  U\subseteq\Cl(S).
\]
\end{theorem}

\begin{proof}
We first prove the sentence case. If \(S\vdash_{\C}\varphi\), then Lemma~\ref{lem:truth-region-soundness}, applied to the indexed truth space \(\F_{\C}\), gives
\[
  V_{\F_{\C}}(S)\subseteq V_{\F_{\C}}(\varphi).
\]
Conversely, assume
\[
  V_{\F_{\C}}(S)\subseteq V_{\F_{\C}}(\varphi).
\]
By idempotence, \(\Cl(S)\) is closed, so \(\Cl(S)\in\Fix(\Cl)\). By extensiveness, \(S\subseteq\Cl(S)\), hence \(\Cl(S)\in V_{\F_{\C}}(S)\). The assumed inclusion gives \(\Cl(S)\in V_{\F_{\C}}(\varphi)\). Since the assignment in \(\F_{\C}\) is the identity, this means \(\varphi\in\Cl(S)\), i.e. \(S\vdash_{\C}\varphi\).

For the general statement, first assume \(V_{\F_{\C}}(S)\subseteq V_{\F_{\C}}(U)\). For each \(\psi\in U\), we have \(V_{\F_{\C}}(U)\subseteq V_{\F_{\C}}(\psi)\), hence \(V_{\F_{\C}}(S)\subseteq V_{\F_{\C}}(\psi)\). By the sentence case, \(\psi\in\Cl(S)\). Therefore \(U\subseteq\Cl(S)\).

Conversely, assume \(U\subseteq\Cl(S)\), and let \(T\in V_{\F_{\C}}(S)\). Then \(T\in\Fix(\Cl)\) and \(S\subseteq T\). By monotonicity,
\[
  \Cl(S)\subseteq\Cl(T)=T.
\]
Since \(U\subseteq\Cl(S)\), it follows that \(U\subseteq T\), so \(T\in V_{\F_{\C}}(U)\). Therefore \(V_{\F_{\C}}(S)\subseteq V_{\F_{\C}}(U)\).
\end{proof}

\begin{example}[Full space: a separating closed theory]
\label{ex:p-implies-q-full}
Let
\[
  \Sig=\{p,q\}.
\]
Let \(\Cl\) be the closure operator generated by the single rule \(p\vdash q\). Thus
\[
  \Cl(\varnothing)=\varnothing,
  \qquad
  \Cl(\{q\})=\{q\},
  \qquad
  \Cl(\{p\})=\Cl(\{p,q\})=\{p,q\}.
\]
The closed theories are
\[
  T_0=\varnothing,
  \qquad
  T_1=\{q\},
  \qquad
  T_2=\{p,q\}.
\]
In the full indexed truth space these three closed theories are all used as indices, and the truth regions are
\[
  V(p)=\{T_2\},
  \qquad
  V(q)=\{T_1,T_2\}.
\]
Hence
\[
  V(p)\subsetneq V(q),
\]
which represents the valid closure consequence \(p\vdash q\). The converse containment fails because
\[
  T_1\in V(q)\setminus V(p).
\]
This single closed theory is the separating witness for \(q\nvdash p\). The example is therefore the smallest useful model of the full representation theorem: entailment is region inclusion, and non-entailment is witnessed by a closed theory in the full index set.
\end{example}

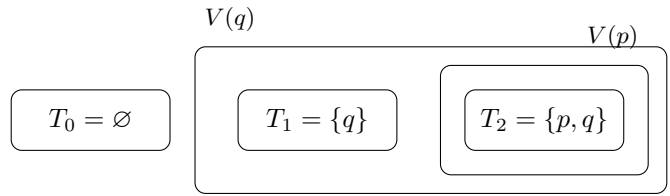
\begin{figure}[!htbp]
\centering
\begin{tikzpicture}[
  theory/.style={draw, rounded corners, minimum width=2.1cm, minimum height=8mm, align=center},
  region/.style={draw, rounded corners},
  every node/.style={font=\small}
]
\node[theory] (T0) at (0,0) {$T_0=\varnothing$};
\node[theory] (T1) at (3,0) {$T_1=\{q\}$};
\node[theory] (T2) at (6,0) {$T_2=\{p,q\}$};

\node[region, fit=(T2), inner sep=9pt] (Vp) {};
\node[region, fit=(T1)(T2), inner xsep=16pt, inner ysep=16pt] (Vq) {};

\node[font=\footnotesize, above=2pt of Vq.north west, anchor=south west] {$V(q)$};
\node[font=\footnotesize, above=2pt of Vp.north east, anchor=south east] {$V(p)$};

\node[font=\footnotesize] at (3,-1.62)
  {$V(p)\subsetneq V(q)$, with $T_1\in V(q)\setminus V(p)$.};
\end{tikzpicture}
\caption{Truth regions for the closure rule \(p\vdash q\).}
\label{fig:full-truth-region-pq}
\end{figure}

\begin{counterexample}[Reduced space and spurious region consequence]
\label{cex:reduced-index-spurious-region-consequence}
Use the same closure system as in Example~\ref{ex:p-implies-q-full}, but omit the separating closed theory \(T_1=\{q\}\).  Thus the reduced indexed truth space has
\[
  I'=\{T_0,T_2\}=\{\varnothing,\{p,q\}\}.
\]
Inside this reduced space,
\[
  V_{I'}(p)=\{T_2\}=V_{I'}(q).
\]
Consequently
\[
  \{q\}\models_{I'}p.
\]
This is a \emph{spurious region consequence}: it is true as a statement about the reduced index family, but it is not a closure consequence of the underlying system, since
\[
  \{q\}\nvdash p.
\]
The reason is exactly visible: the omitted closed theory \(T_1=\{q\}\) is the witness that \(q\) can hold without \(p\).  Thus reduced indexed spaces can collapse distinct truth regions and create false local entailments relative to the original closure operator.
\end{counterexample}

\begin{remark}[Transition to decision data]
This section studies positive truth regions only.  It does not yet compare a sentence with a designated opposite, nor does it impose consistency or inconsistency conditions on indexed theories.  The next section adds exactly this bookkeeping layer.  The result remains finite and structural: truth regions are still subsets of the chosen index set, not open sets in an assumed topology.
\end{remark}

\section{Decision Decompositions}
\label{sec:decision-decompositions}

The representation theorem above concerns positive occurrence of sentences in indexed closed theories.  To discuss decision status one must choose, in addition, which sentence is being treated as the opposite of a given sentence.  This section adds only that bookkeeping layer.  It is auxiliary to the later atlas theorem: the proof of the local-to-global obstruction criterion uses closure propagation and conservative restriction, not decision signatures or decision distances.  The point of the section is more limited and more precise: truth-region language should not silently import a semantics of negation.

Accordingly, the notation below is intentionally weak.  A designated symbol \(\bar a\) is simply another sentence whose truth region is compared with the truth region of \(a\).  Any exclusivity, inconsistency, explosion, or classical behavior must come either from the closure operator itself or from an explicitly imposed condition.

\begin{definition}[Formal opposites and four-region decision decomposition]
\label{def:four-region-decomposition}
Let \(\X=(\Sig,\Cl,I,\tau)\) be a finite indexed truth space.  For a chosen sentence \(a\in\Sig\), a \emph{formal opposite} of \(a\) is a designated sentence \(\bar a\in\Sig\).  No inference rule relating \(a\) and \(\bar a\) is assumed unless it is explicitly present in the closure operator or imposed as an additional condition.

Given such a pair \((a,\bar a)\), define
\[
  P_a=V_{\X}(a)\setminus V_{\X}(\bar a),
  \qquad
  N_a=V_{\X}(\bar a)\setminus V_{\X}(a),
\]
\[
  B_a=V_{\X}(a)\cap V_{\X}(\bar a),
  \qquad
  U_a=I\setminus\bigl(V_{\X}(a)\cup V_{\X}(\bar a)\bigr).
\]
The tuple
\[
  \mathcal D_4^{\X}(a)=(P_a,N_a,B_a,U_a)
\]
is the \emph{four-region decision decomposition} of \(a\) in \(\X\).  The four regions record, respectively, the contexts containing \(a\) only, containing \(\bar a\) only, containing both, and containing neither.  When \(\X\) is clear from context, we write simply \(\mathcal D_4(a)\).
\end{definition}

\begin{definition}[Local opposite-separation]
\label{def:opposite-separation}
The indexed truth space \(\X\) is \emph{opposite-separated at \(a\)} if
\[
  B_a=\varnothing.
\]
Equivalently, no indexed closed theory contains both \(a\) and \(\bar a\).  For a family of formal-opposite pairs, \(\X\) is \emph{opposite-separated} if this condition holds for every pair under consideration.
\end{definition}

\begin{remark}[Partition and the opposite-separated reduction]
The definitions give the disjoint partition
\[
I=P_a\sqcup N_a\sqcup B_a\sqcup U_a,
\]
and
\[
V_{\mathfrak X}(a)=P_a\sqcup B_a,
\qquad
V_{\mathfrak X}(\bar a)=N_a\sqcup B_a.
\]

\Needspace{8\baselineskip}
The four-region decomposition can be displayed as the following membership table:
\begin{center}
\renewcommand{\arraystretch}{1.18}
\begin{tabular}{@{}c|cc@{}}
 & $\bar a\in\tau(i)$ & $\bar a\notin\tau(i)$ \\
\midrule
$a\in\tau(i)$
  & $\begin{array}{c} B_a \\[-2pt] \scriptstyle a\text{ and }\bar a \end{array}$
  & $\begin{array}{c} P_a \\[-2pt] \scriptstyle a\text{ only} \end{array}$ \\
$a\notin\tau(i)$
  & $\begin{array}{c} N_a \\[-2pt] \scriptstyle \bar a\text{ only} \end{array}$
  & $\begin{array}{c} U_a \\[-2pt] \scriptstyle \text{neither} \end{array}$ \\
\end{tabular}
\end{center}

When \(\mathfrak X\) is opposite-separated at \(a\), the region \(B_a\) is empty,
so the decomposition reduces to the usual positive--negative--undecided
trichotomy.
\end{remark}

\begin{definition}[Optional finite summaries]
\label{def:decision-summaries}
The \emph{four-valued decision state function} of \(a\) is
\[
  \delta_a^{(4)}:I\to\{+,-,\pm,0\},
  \qquad
  \delta_a^{(4)}(i)=
  \begin{cases}
  + & i\in P_a,\\
  - & i\in N_a,\\
  \pm & i\in B_a,\\
  0 & i\in U_a.
  \end{cases}
\]
The symbol \(\pm\) records simultaneous occurrence of \(a\) and \(\bar a\) in the assigned closed theory.  It does not specify any paraconsistent or explosive inference rule.

The \emph{four-region decision signature} of \(a\) is
\[
  \sigma_4(a)=\bigl(|P_a|,|N_a|,|B_a|,|U_a|\bigr).
\]
If \(I\neq\varnothing\), its normalized version is
\[
  \bar\sigma_4(a)=
  \left(
  \frac{|P_a|}{|I|},
  \frac{|N_a|}{|I|},
  \frac{|B_a|}{|I|},
  \frac{|U_a|}{|I|}
  \right).
\]
For selected sentences \(a,b\in\Sig\) with designated formal opposites \(\bar a,\bar b\), define
\[
  d_4(a,b)=\bigl|\{i\in I: \delta_a^{(4)}(i)\neq\delta_b^{(4)}(i)\}\bigr|.
\]
These are finite summaries of \(\mathcal D_4\), not replacements for the underlying regions.  They are not used in the obstruction theorem.
\end{definition}

\begin{remark}[Pseudometric status of \(d_4\)]
The quantity \(d_4\) is the ordinary Hamming distance between the state functions \(\delta_a^{(4)}\) and \(\delta_b^{(4)}\), pulled back along the assignment \(a\mapsto\delta_a^{(4)}\).  Hence it is a pseudometric on the selected sentences equipped with formal opposites.  It may vanish on distinct sentences, and
\[
  d_4(a,b)=0
  \quad\Longleftrightarrow\quad
  \delta_a^{(4)}=\delta_b^{(4)}
  \quad\Longleftrightarrow\quad
  \mathcal D_4(a)=\mathcal D_4(b).
\]
This is a bookkeeping fact, not a structural input to the atlas construction.
\end{remark}

\begin{example}[A nonempty both-region]
\label{ex:nonempty-both-region}
Let \(\Sig=\{a,\bar a\}\), and let \(\Cl(S)=S\) be the identity closure.  The full indexed truth space has one index for each subset of \(\Sig\):
\[
  I=\Pow(\Sig)=\{\varnothing,\{a\},\{\bar a\},\{a,\bar a\}\}.
\]
The truth regions are
\[
  V(a)=\{\{a\},\{a,\bar a\}\},
  \qquad
  V(\bar a)=\{\{\bar a\},\{a,\bar a\}\}.
\]
Therefore
\[
  P_a=\{\{a\}\},\quad
  N_a=\{\{\bar a\}\},\quad
  B_a=\{\{a,\bar a\}\},\quad
  U_a=\{\varnothing\}.
\]
The both-region is nonempty because neither the closure system nor the index family prevents simultaneous occurrence of the two designated sentences.
\end{example}

\begin{example}[Signatures do not determine decompositions]
\label{ex:signature-not-complete}
Let \(\Sig=\{a,\bar a,b,\bar b\}\) with identity closure, and take two indices \(i_1,i_2\) with
\[
  \tau(i_1)=\{a,\bar b\},
  \qquad
  \tau(i_2)=\{\bar a,b\}.
\]
For \(a\), one has \(P_a=\{i_1\}\) and \(N_a=\{i_2\}\).  For \(b\), one has \(P_b=\{i_2\}\) and \(N_b=\{i_1\}\).  In both cases the both-region and undecided region are empty, so
\[
  \sigma_4(a)=\sigma_4(b)=(1,1,0,0).
\]
Nevertheless the two state functions disagree at both indices, and therefore \(d_4(a,b)=2\).  Thus the signature records only the sizes of the four regions; it does not determine their positions in the common index set.
\end{example}

\begin{remark}[Role of decision decompositions]
The essential content of this section is the four-region partition.  The optional summaries above can be useful when comparing finite region-level behavior inside a fixed indexed truth space, but they are secondary to the actual regions and play no role in the proof of the local-to-global obstruction criterion.  With this bookkeeping separated from closure propagation, the next section turns to conservative comparison between finite closure systems.
\end{remark}

\section{Conservative Extensions and Truth-Region Pullback}
\label{sec:conservative-closure-extensions}

The previous sections treated the internal structure of a single indexed truth space.  We now record the comparison principle needed before several local systems can be assembled into an atlas.  The intended situation is that a larger language may introduce new sentences and new consequences involving those sentences, while leaving the consequences expressible in an older language unchanged.

This is the role of conservative extension.  In full truth spaces, conservativity has a precise truth-region form: old truth regions in the larger system are exact inverse images of old truth regions in the smaller system under the reduct map.  This one-extension result will become the chart-by-chart requirement in the atlas theorem.

\begin{definition}[Conservative extension of finite closure systems]
\label{def:conservative-closure-extension}
Let
\[
  \C_X=(\Sig_X,\Cl_X),
  \qquad
  \C_Y=(\Sig_Y,\Cl_Y)
\]
be finite closure systems with \(\Sig_X\subseteq\Sig_Y\). We say that \(\C_Y\) is a \emph{conservative extension} of \(\C_X\) if, for every \(S\subseteq\Sig_X\),
\[
  \Cl_Y(S)\cap\Sig_X=\Cl_X(S).
\]
Thus closing an old-language set inside the larger system and then restricting back to the old language gives exactly the same old consequences as closing it inside the old system.
\end{definition}

\begin{lemma}[Preservation of old consequence]
\label{lem:conservative-extension-preserves-old-consequence}
If \(\C_Y\) is a conservative extension of \(\C_X\), then for every \(S\subseteq\Sig_X\) and every \(\varphi\in\Sig_X\),
\[
  S\vdash_{\C_X}\varphi
  \quad\Longleftrightarrow\quad
  S\vdash_{\C_Y}\varphi.
\]
\end{lemma}

\begin{proof}
By definition,
\[
  S\vdash_{\C_X}\varphi
  \quad\Longleftrightarrow\quad
  \varphi\in\Cl_X(S).
\]
Since the extension is conservative,
\[
  \Cl_X(S)=\Cl_Y(S)\cap\Sig_X.
\]
Because \(\varphi\in\Sig_X\), membership in \(\Cl_X(S)\) is equivalent to membership in \(\Cl_Y(S)\). This is exactly the asserted equivalence.
\end{proof}

For \(T\subseteq\Sig_Y\), write
\[
  T\restr{X}:=T\cap\Sig_X
\]
for its old-language reduct.  This is only set-theoretic restriction; the following proposition explains why it sends closed theories to closed theories under conservativity.

\begin{lemma}[Reducts of closed theories are closed]
\label{lem:reducts-of-closed-theories-are-closed}
Assume that \(\C_Y\) is a conservative extension of \(\C_X\). If
\[
  T\in\Fix(\Cl_Y),
\]
then
\[
  T\cap\Sig_X\in\Fix(\Cl_X).
\]
\end{lemma}

\begin{proof}
Let \(T_X=T\cap\Sig_X\). We show that \(\Cl_X(T_X)=T_X\). By conservativity,
\[
  \Cl_X(T_X)=\Cl_Y(T_X)\cap\Sig_X.
\]
Since \(T_X\subseteq T\), monotonicity gives
\[
  \Cl_Y(T_X)\subseteq\Cl_Y(T)=T,
\]
because \(T\) is \(\Cl_Y\)-closed. Therefore
\[
  \Cl_X(T_X)=\Cl_Y(T_X)\cap\Sig_X\subseteq T\cap\Sig_X=T_X.
\]
The reverse inclusion follows from extensiveness of \(\Cl_X\). Hence \(\Cl_X(T_X)=T_X\), so \(T_X\) is \(\Cl_X\)-closed.
\end{proof}

\begin{definition}[Reduct map on full truth spaces]
\label{def:reduct-map-full-truth-spaces}
Let
\[
  \F_X=(\Sig_X,\Cl_X,\Fix(\Cl_X),\mathrm{id}),
  \qquad
  \F_Y=(\Sig_Y,\Cl_Y,\Fix(\Cl_Y),\mathrm{id})
\]
be the full indexed truth spaces of \(\C_X\) and \(\C_Y\). If \(\C_Y\) is conservative over \(\C_X\), define
\[
  \rho:\Fix(\Cl_Y)\to\Fix(\Cl_X),
  \qquad
  \rho(T)=T\cap\Sig_X.
\]
This map is well-defined by Lemma~\ref{lem:reducts-of-closed-theories-are-closed}.
\end{definition}

\begin{proposition}[The reduct map is surjective]
\label{prop:reduct-map-surjective}
If \(\C_Y\) is a conservative extension of \(\C_X\), then the reduct map
\[
  \rho:\Fix(\Cl_Y)\to\Fix(\Cl_X)
\]
is surjective.
\end{proposition}

\begin{proof}
Let \(C\in\Fix(\Cl_X)\). Define
\[
  \widehat C=\Cl_Y(C).
\]
Then \(\widehat C\in\Fix(\Cl_Y)\) by idempotence. Its reduct is
\[
  \rho(\widehat C)=\Cl_Y(C)\cap\Sig_X.
\]
By conservativity,
\[
  \Cl_Y(C)\cap\Sig_X=\Cl_X(C).
\]
Since \(C\) is \(\Cl_X\)-closed, \(\Cl_X(C)=C\). Hence \(\rho(\widehat C)=C\). Therefore every old closed theory is the reduct of some new closed theory.
\end{proof}

\begin{theorem}[Truth-region pullback under conservative extension]
\label{thm:old-truth-regions-pull-back-conservative-extension}
Assume that \(\C_Y\) is a conservative extension of \(\C_X\), and let
\[
  \rho:\Fix(\Cl_Y)\to\Fix(\Cl_X)
\]
be the reduct map. For every old-language sentence \(\varphi\in\Sig_X\),
\[
  V_{\F_Y}(\varphi)=\rho^{-1}\bigl(V_{\F_X}(\varphi)\bigr).
\]
More generally, for every \(S\subseteq\Sig_X\),
\[
  V_{\F_Y}(S)=\rho^{-1}\bigl(V_{\F_X}(S)\bigr).
\]
\end{theorem}

\begin{proof}
Let \(T\in\Fix(\Cl_Y)\). First let \(\varphi\in\Sig_X\). Then
\[
  T\in V_{\F_Y}(\varphi)
  \quad\Longleftrightarrow\quad
  \varphi\in T.
\]
Since \(\varphi\in\Sig_X\), this is equivalent to
\[
  \varphi\in T\cap\Sig_X=\rho(T).
\]
Equivalently,
\[
  T\in V_{\F_Y}(\varphi)
  \quad\Longleftrightarrow\quad
  \rho(T)\in V_{\F_X}(\varphi).
\]
This proves the sentence case.

For \(S\subseteq\Sig_X\), the same argument gives
\[
  T\in V_{\F_Y}(S)
  \quad\Longleftrightarrow\quad
  S\subseteq T
  \quad\Longleftrightarrow\quad
  S\subseteq T\cap\Sig_X
  \quad\Longleftrightarrow\quad
  \rho(T)\in V_{\F_X}(S).
\]
Thus \(V_{\F_Y}(S)=\rho^{-1}(V_{\F_X}(S))\).
\end{proof}

\begin{corollary}[Old full region consequence is preserved]
\label{cor:old-full-region-consequence-preserved}
Assume that \(\C_Y\) is a conservative extension of \(\C_X\). For old-language \(S\subseteq\Sig_X\) and \(\varphi\in\Sig_X\),
\[
  V_{\F_X}(S)\subseteq V_{\F_X}(\varphi)
  \quad\Longleftrightarrow\quad
  V_{\F_Y}(S)\subseteq V_{\F_Y}(\varphi).
\]
\end{corollary}

\begin{proof}
By Theorem~\ref{thm:full-truth-region-representation}, the first truth-region inclusion is equivalent to \(S\vdash_{\C_X}\varphi\), and the second is equivalent to \(S\vdash_{\C_Y}\varphi\). These two consequence statements are equivalent by Lemma~\ref{lem:conservative-extension-preserves-old-consequence}.
\end{proof}

\begin{example}[Conservative extension adding only a new-language consequence]
\label{ex:conservative-extension-new-language-consequence}
Let \(\Sig_X=\{p\}\), with identity closure \(\Cl_X(S)=S\).  Let \(\Sig_Y=\{p,q\}\), and define \(\Cl_Y\) by the single rule \(p\vdash_Y q\).  Thus
\[
  \Cl_Y(\varnothing)=\varnothing,
  \qquad
  \Cl_Y(\{q\})=\{q\},
  \qquad
  \Cl_Y(\{p\})=\Cl_Y(\{p,q\})=\{p,q\}.
\]
For every \(S\subseteq\Sig_X\), closing in the larger system and restricting back to \(\Sig_X\) gives
\[
  \Cl_Y(S)\cap\Sig_X=S=\Cl_X(S).
\]
Hence \(\C_Y\) is conservative over \(\C_X\), even though it adds the new-language consequence \(p\vdash_Y q\).

The full closed theories of \(\C_X\) are \(\varnothing\) and \(\{p\}\).  The full closed theories of \(\C_Y\) are
\[
  \varnothing,
  \qquad
  \{q\},
  \qquad
  \{p,q\}.
\]
The reduct map is
\[
  \rho(\varnothing)=\varnothing,
  \qquad
  \rho(\{q\})=\varnothing,
  \qquad
  \rho(\{p,q\})=\{p\}.
\]
\Needspace{7\baselineskip}
Equivalently, its fibers are summarized by
\begin{center}
\small
\renewcommand{\arraystretch}{1.12}
\begin{tabular}{@{}c c c@{}}
\toprule
$T\in\Fix(\Cl_Y)$ & & $\rho(T)\in\Fix(\Cl_X)$ \\
\midrule
$\varnothing$ & $\mapsto$ & $\varnothing$ \\
$\{q\}$ & $\mapsto$ & $\varnothing$ \\
$\{p,q\}$ & $\mapsto$ & $\{p\}$ \\
\bottomrule
\end{tabular}
\end{center}
For the old sentence \(p\),
\[
  V_{\F_X}(p)=\{\{p\}\},
  \qquad
  V_{\F_Y}(p)=\{\{p,q\}\}.
\]
Thus
\[
  V_{\F_Y}(p)=\rho^{-1}\bigl(V_{\F_X}(p)\bigr).
\]
This illustrates the point of conservativity: new sentences and new mixed-language consequences may appear, but old-language truth regions pull back exactly.
\end{example}
\begin{counterexample}[A nonconservative extension creates an old consequence]
\label{cex:nonconservative-extension-creates-old-consequence}
Let \(\Sig_X=\{p,q\}\) with identity closure, and let \(\Sig_Y=\Sig_X\) with closure generated by the rule \(p\vdash_Y q\).  Then
\[
  \Cl_X(\{p\})=\{p\},
  \qquad
  \Cl_Y(\{p\})\cap\Sig_X=\{p,q\}.
\]
So \(\C_Y\) is not conservative over \(\C_X\).  The old-language consequence
\[
  p\vdash_Y q
\]
has been created by the extension, although
\[
  p\nvdash_X q.
\]
Equivalently, the extension has created the old truth-region containment corresponding to \(p\vdash q\).  This is precisely the behavior excluded by the pullback theorem.
\end{counterexample}

\begin{remark}[Cardinality is not preserved by pullback]
The equality
\[
  V_{\F_Y}(\varphi)=\rho^{-1}(V_{\F_X}(\varphi))
\]
preserves the old truth region structurally, but it need not preserve cardinal summaries. The reduct map may have fibers of different sizes. In Example~\ref{ex:conservative-extension-new-language-consequence},
\[
  \rho^{-1}(\varnothing)=\{\varnothing,\{q\}\}
  \quad\text{but}\quad
  \rho^{-1}(\{p\})=\{\{p,q\}\}.
\]
Thus cardinality-based summaries, including decision signatures, depend on the ambient indexed space or on an explicit measure on it.
\end{remark}

\begin{remark}[Scope of the comparison theory]
General refinement maps between selected indexed families are useful for comparing chosen subspaces, but the present manuscript only needs the closure-level result above: conservative extension preserves old consequence and induces exact pullback on old truth regions in full truth spaces.  The next section replaces the single inclusion \(\Sig_X\subseteq\Sig_Y\) by several overlapping chart languages.  Conservative realization will then mean that every chart is preserved in the same sense simultaneously.
\end{remark}

\section{Closure Atlases and Atlas-Generated Closure}
\label{sec:closure-atlases}

We now pass from a single conservative extension to a finite family of local closure systems.  The purpose is to study when these local consequence relations can be regarded as restrictions of one global consequence relation.  The local systems may have overlapping sentence universes, and their interactions can force consequences not visible from any one chart alone.  This section constructs the canonical global closure generated by such interactions; the next section asks when that canonical closure is conservative over each chart.

\begin{definition}[Finite closure atlas]
A \emph{finite closure atlas} is a finite family of finite closure systems
\[
  \A=\{(\Sig_\alpha,\Cl_\alpha)\}_{\alpha\in A}.
\]
The finite set \(A\) is the \emph{chart index set}. The \emph{global sentence universe} of the atlas is
\[
  \Sig_{\A}=\bigcup_{\alpha\in A}\Sig_\alpha.
\]
When the atlas is clear from context, we write simply \(\Sig\) for \(\Sig_{\A}\). For charts \(\alpha,\beta\in A\), their \emph{overlap language} is
\[
  \Sig_{\alpha\beta}=\Sig_\alpha\cap\Sig_\beta.
\]
\end{definition}

\begin{remark}[No topology or prior compatibility is assumed]
The word ``atlas'' is used only to indicate a finite family of local closure systems with possibly overlapping languages. No topological structure is assumed. Moreover, the definition does not require any prior agreement on overlaps. If two charts disagree on an overlap, or if several charts jointly force a consequence that one chart denies, this will be detected later by the obstruction set.
\end{remark}

\begin{definition}[Atlas step operator and atlas-generated closure]
For \(X\subseteq\Sig_{\A}\), define the \emph{atlas step operator}
\[
  \GA:\Pow(\Sig_{\A})\to\Pow(\Sig_{\A})
\]
by
\[
  \GA(X)=X\cup\bigcup_{\alpha\in A}\Cl_\alpha(X\cap\Sig_\alpha).
\]
Thus \(\GA(X)\) is obtained by allowing each chart \(\alpha\) to close the part of \(X\) visible in its own sentence universe \(\Sig_\alpha\).

Define \(\GA^0(X)=X\) and, recursively,
\[
  \GA^{n+1}(X)=\GA(\GA^n(X)).
\]
Since \(\Sig_{\A}\) is finite, the increasing sequence
\[
  \GA^0(X)\subseteq\GA^1(X)\subseteq\GA^2(X)\subseteq\cdots
\]
stabilizes after finitely many steps. The \emph{atlas-generated closure} of \(X\) is
\[
  \CA(X)=\bigcup_{n\geq0}\GA^n(X),
\]
equivalently the stable value of the iteration of \(\GA\) starting from \(X\).
\end{definition}

\begin{lemma}[Monotonicity of the atlas step operator]
\label{lem:atlas-step-monotonicity}
If \(X\subseteq Y\subseteq\Sig_{\A}\), then
\[
  \GA(X)\subseteq\GA(Y).
\]
\end{lemma}

\begin{proof}
If \(X\subseteq Y\), then \(X\cap\Sig_\alpha\subseteq Y\cap\Sig_\alpha\) for each chart \(\alpha\).  Monotonicity of \(\Cl_\alpha\) gives
\[
  \Cl_\alpha(X\cap\Sig_\alpha)
  \subseteq
  \Cl_\alpha(Y\cap\Sig_\alpha),
\]
and the defining formula for \(\GA\) then yields \(\GA(X)\subseteq\GA(Y)\).
\end{proof}

\begin{proposition}[Atlas-generated closure]
\label{prop:atlas-generated-closure-is-closure}
The map
\[
  \CA:\Pow(\Sig_{\A})\to\Pow(\Sig_{\A})
\]
is a closure operator.
\end{proposition}

\begin{proof}
Extensiveness is immediate from \(\GA^0(X)=X\).  If \(X\subseteq Y\), Lemma~\ref{lem:atlas-step-monotonicity} gives \(\GA^n(X)\subseteq\GA^n(Y)\) for every \(n\), so taking unions gives \(\CA(X)\subseteq\CA(Y)\).

For idempotence, finiteness of \(\Sig_{\A}\) implies that the increasing sequence \(\GA^0(X)\subseteq\GA^1(X)\subseteq\cdots\) stabilizes.  Choose \(N\) with
\[
  \CA(X)=\GA^N(X)=\GA^{N+1}(X).
\]
Then \(\GA(\CA(X))=\CA(X)\).  Hence every iterate starting from \(\CA(X)\) is again \(\CA(X)\), and therefore
\[
  \CA(\CA(X))=\CA(X).
\]
Thus \(\CA\) is a closure operator.
\end{proof}

\begin{definition}[Global extension of local chart closures]
\label{def:global-extension-chart-closures}
Let
\[
  D:\Pow(\Sig_{\A})\to\Pow(\Sig_{\A})
\]
be a closure operator. We say that \(D\) \emph{extends the local chart closures} if for every chart \(\alpha\in A\) and every \(S\subseteq\Sig_\alpha\),
\[
  \Cl_\alpha(S)\subseteq D(S).
\]
This condition says that every local consequence remains valid globally. It is weaker than conservative realization, which is introduced in the next section.
\end{definition}

\begin{theorem}[Minimality of the atlas-generated closure]
\label{thm:atlas-generated-closure-minimal}
The atlas-generated closure \(\CA\) is the least global closure operator extending all local chart closures. More explicitly:
\begin{enumerate}[label=\textup{(\roman*)}]
  \item \(\CA\) extends each local chart closure;
  \item if \(D:\Pow(\Sig_{\A})\to\Pow(\Sig_{\A})\) is any closure operator extending all local chart closures, then for every \(X\subseteq\Sig_{\A}\),
  \[
    \CA(X)\subseteq D(X).
  \]
\end{enumerate}
\end{theorem}

\begin{proof}
First fix a chart \(\alpha\in A\) and a local set \(S\subseteq\Sig_\alpha\). Since \(S\cap\Sig_\alpha=S\), the definition of \(\GA\) gives
\[
  \Cl_\alpha(S)=\Cl_\alpha(S\cap\Sig_\alpha)
  \subseteq \GA(S)
  \subseteq \CA(S).
\]
Thus \(\CA\) extends every local chart closure.

Now let \(D\) be any closure operator on \(\Sig_{\A}\) extending every local chart closure. We prove by induction on \(n\) that
\[
  \GA^n(X)\subseteq D(X)
\]
for every \(n\geq0\).

For \(n=0\), this is just \(X\subseteq D(X)\), which follows from extensiveness of \(D\). Suppose the inclusion holds at stage \(n\), and put
\[
  S_\alpha=\GA^n(X)\cap\Sig_\alpha.
\]
Since \(D\) extends local chart closures,
\[
  \Cl_\alpha(S_\alpha)\subseteq D(S_\alpha).
\]
By the induction hypothesis,
\[
  S_\alpha\subseteq \GA^n(X)\subseteq D(X).
\]
Monotonicity and idempotence of \(D\) give
\[
  D(S_\alpha)\subseteq D(D(X))=D(X).
\]
Hence \(\Cl_\alpha(S_\alpha)\subseteq D(X)\) for every \(\alpha\). Combining this with \(\GA^n(X)\subseteq D(X)\), we obtain
\[
  \GA^{n+1}(X)
  =\GA^n(X)\cup\bigcup_{\alpha\in A}\Cl_\alpha(S_\alpha)
  \subseteq D(X).
\]
The induction is complete. Taking the union over all \(n\) yields
\[
  \CA(X)=\bigcup_{n\geq0}\GA^n(X)\subseteq D(X).
\]
Together with the first paragraph and Proposition~\ref{prop:atlas-generated-closure-is-closure}, this proves that \(\CA\) is the least global closure operator extending the local chart closures.
\end{proof}

\begin{remark}[Extension versus conservative realization]
The minimality theorem concerns extension of local consequence:
\[
  \Cl_\alpha(S)\subseteq D(S).
\]
A global conservative realization satisfies the stronger condition
\[
  D(S)\cap\Sig_\alpha=\Cl_\alpha(S)
  \qquad(S\subseteq\Sig_\alpha).
\]
Thus \(\CA\) always extends the local charts, but it may fail to be conservative over them. The obstruction set introduced in the next section measures exactly this possible failure.
\end{remark}

\begin{example}[Propagation through overlapping charts without obstruction]
\label{ex:propagation-through-charts}
Let
\[
  \Sig_1=\{p,q\},\qquad \Sig_2=\{q,r\}.
\]
Let \(\Cl_1\) be generated by \(p\vdash_1 q\), and let \(\Cl_2\) be generated by \(q\vdash_2 r\).  Starting from
\[
  X_0=\{p\},
\]
the atlas iteration gives
\[
  X_1=\GA(X_0)=\{p,q\},
\]
because chart \(1\) sees \(p\) and adds \(q\).  Then
\[
  X_2=\GA(X_1)=\{p,q,r\},
\]
because chart \(2\) sees \(q\) and adds \(r\).  The process now stabilizes, so
\[
  \CA(\{p\})=\{p,q,r\}.
\]
Thus the atlas-generated closure forces the global consequence \(p\vdash r\), although no single chart contains both the premise \(p\) and the conclusion \(r\) together with a direct rule \(p\vdash r\).  The consequence is produced by propagation through the overlap sentence \(q\).  By minimality, every global closure operator extending both chart closures must contain this propagated consequence.

This propagated consequence is not yet an obstruction.  For chart \(1\), the four local subsets give
\[
\begin{array}{c|c|c}
S & \CA(S)\cap\Sig_1 & \Cl_1(S)\\
\hline
\varnothing & \varnothing & \varnothing\\
\{p\} & \{p,q\} & \{p,q\}\\
\{q\} & \{q\} & \{q\}\\
\{p,q\} & \{p,q\} & \{p,q\}
\end{array}
\]
and for chart \(2\), the corresponding equalities are
\[
\begin{array}{c|c|c}
S & \CA(S)\cap\Sig_2 & \Cl_2(S)\\
\hline
\varnothing & \varnothing & \varnothing\\
\{q\} & \{q,r\} & \{q,r\}\\
\{r\} & \{r\} & \{r\}\\
\{q,r\} & \{q,r\} & \{q,r\}.
\end{array}
\]
Therefore \(E_{\alpha,S}=\varnothing\) for every chart \(\alpha\) and every local premise set \(S\subseteq\Sig_\alpha\).  Hence \(\Obs(\A)=\varnothing\), and the atlas-generated closure \(\CA\) is itself a global conservative realization.  The point is that obstruction is not propagation itself; obstruction occurs only when propagated information returns as a chart-visible consequence not already validated by that chart.

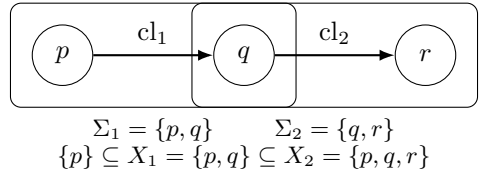
\begin{figure}[!htbp]
\centering
\begin{tikzpicture}[
  sent/.style={circle, draw, minimum size=8mm},
  rule/.style={-{Latex[length=2mm]}, thick},
  chart/.style={draw, rounded corners, inner sep=8pt},
  every node/.style={font=\small}
]
\node[sent] (p) at (0,0) {$p$};
\node[sent] (q) at (2.4,0) {$q$};
\node[sent] (r) at (4.8,0) {$r$};

\draw[rule] (p) -- node[above] {$\Cl_1$} (q);
\draw[rule] (q) -- node[above] {$\Cl_2$} (r);

\node[chart, fit=(p)(q), label={[font=\footnotesize]below:$\Sigma_1=\{p,q\}$}] {};
\node[chart, fit=(q)(r), label={[font=\footnotesize]below:$\Sigma_2=\{q,r\}$}] {};

\node[font=\footnotesize] at (2.4,-1.35)
  {$\{p\}\subseteq X_1=\{p,q\}\subseteq X_2=\{p,q,r\}$};
\end{tikzpicture}
\caption{Atlas propagation without obstruction.}
\label{fig:atlas-propagation-no-obstruction}
\end{figure}

\end{example}

The example illustrates why the atlas-generated closure is the canonical object: it contains exactly the consequences forced by repeated local propagation.  It also shows that a propagated global consequence need not be an obstruction.  The remaining question is whether any forced consequence becomes visible inside a chart where it was not locally valid.  That question is answered by the obstruction criterion.

\section{Local-to-Global Obstructions}
\label{sec:obstructions}

The atlas-generated closure \(\CA\) is the least global closure operator extending all local chart closures.  Extension alone is not enough for a faithful globalization: a global closure may preserve every local implication while also creating a new implication whose premises and conclusion already lie in one chart language.  Such a consequence is visible locally, but not locally valid.  The obstruction set records exactly these chart-visible excess consequences.

\begin{definition}[Global conservative realization]
Let
\[
  \A=\{(\Sig_\alpha,\Cl_\alpha)\}_{\alpha\in A}
\]
be a finite closure atlas, with global sentence universe \(\Sig_{\A}\).  A closure operator
\[
  D:\Pow(\Sig_{\A})\to\Pow(\Sig_{\A})
\]
is a \emph{global conservative realization} of \(\A\) if for every chart \(\alpha\in A\) and every local premise set \(S\subseteq\Sig_\alpha\),
\[
  D(S)\cap\Sig_\alpha=\Cl_\alpha(S).
\]
Equivalently, the local consequence relation of chart \(\alpha\) is recovered exactly by restricting the global closure back to \(\Sig_\alpha\).  This is stronger than extension: extension requires only \(\Cl_\alpha(S)\subseteq D(S)\), while conservative realization also forbids new \(\Sig_\alpha\)-sentences from appearing in \(D(S)\).
\end{definition}

\begin{definition}[Chart excess and obstruction set]
\label{def:atlas-obstruction-set}
For a chart \(\alpha\in A\) and a local premise set \(S\subseteq\Sig_\alpha\), define the \emph{chart excess} of the atlas-generated closure over chart \(\alpha\) at \(S\) by
\[
  E_{\alpha,S}
  =
  \bigl(\CA(S)\cap\Sig_\alpha\bigr)\setminus \Cl_\alpha(S).
\]
Thus \(E_{\alpha,S}\) consists of the \(\alpha\)-visible consequences that are forced by propagation through the whole atlas but are not consequences in chart \(\alpha\) itself.  The \emph{obstruction set} is
\[
  \Obs(\A)
  =
  \{(\alpha,S,\varphi):
      S\subseteq\Sig_\alpha,
      \ \varphi\in E_{\alpha,S}\}.
\]
Equivalently, \((\alpha,S,\varphi)\in\Obs(\A)\) iff
\[
  S\subseteq\Sig_\alpha,
  \qquad
  \varphi\in\Sig_\alpha,
  \qquad
  \varphi\in\CA(S),
  \qquad
  \varphi\notin\Cl_\alpha(S).
\]
\end{definition}

The theorem below is the technical core of the manuscript.  The point is not merely that obstructions prevent conservative globalization.  The canonical closure \(\CA\) is minimal among all global extensions, so every possible global realization must contain the consequences generated by \(\CA\).  Therefore any chart-visible excess produced by \(\CA\) is unavoidable.  Conversely, if \(\CA\) produces no such excess, then the canonical closure itself already realizes the atlas conservatively.

\begin{theorem}[Local-to-global obstruction criterion]
\label{thm:complete-closure-atlas-obstruction}
For a finite closure atlas \(\A\), the following are equivalent:
\begin{enumerate}[label=\textup{(\roman*)}]
  \item \(\A\) has a global conservative realization;
  \item the atlas-generated closure \(\CA\) is a global conservative realization of \(\A\);
  \item \(\Obs(\A)=\varnothing\).
\end{enumerate}
Equivalently, conservative globalizability is exactly the condition that
\[
  \CA(S)\cap\Sig_\alpha=\Cl_\alpha(S)
\]
for every chart \(\alpha\in A\) and every local set \(S\subseteq\Sig_\alpha\).
\end{theorem}

\begin{proof}
The implication \(\textup{(ii)}\Rightarrow\textup{(i)}\) is immediate.  We prove \(\textup{(i)}\Rightarrow\textup{(iii)}\) and \(\textup{(iii)}\Rightarrow\textup{(ii)}\).

Assume first that \(D\) is a global conservative realization.  Then for each \(S\subseteq\Sig_\alpha\),
\[
  \Cl_\alpha(S)=D(S)\cap\Sig_\alpha\subseteq D(S),
\]
so \(D\) extends every local chart closure.  By the minimality of \(\CA\),
\[
  \CA(X)\subseteq D(X)
\]
for every \(X\subseteq\Sig_{\A}\).

Suppose, toward a contradiction, that \((\alpha,S,\varphi)\in\Obs(\A)\).  Then \(S\subseteq\Sig_\alpha\), \(\varphi\in\Sig_\alpha\), \(\varphi\in\CA(S)\), and \(\varphi\notin\Cl_\alpha(S)\).  Since \(\CA(S)\subseteq D(S)\), we have
\[
  \varphi\in D(S)\cap\Sig_\alpha.
\]
Conservativity of \(D\) over chart \(\alpha\) gives
\[
  D(S)\cap\Sig_\alpha=\Cl_\alpha(S),
\]
contradicting \(\varphi\notin\Cl_\alpha(S)\).  Hence \(\Obs(\A)=\varnothing\).

Conversely, assume \(\Obs(\A)=\varnothing\).  By Proposition~\ref{prop:atlas-generated-closure-is-closure}, \(\CA\) is a closure operator on \(\Sig_{\A}\).  It remains only to prove conservative restriction.  Fix \(\alpha\in A\) and \(S\subseteq\Sig_\alpha\).  Since \(\CA\) extends the local chart closures,
\[
  \Cl_\alpha(S)\subseteq \CA(S)\cap\Sig_\alpha.
\]
For the reverse inclusion, if \(\varphi\in\CA(S)\cap\Sig_\alpha\) but \(\varphi\notin\Cl_\alpha(S)\), then \((\alpha,S,\varphi)\in\Obs(\A)\), contradicting emptiness.  Therefore
\[
  \CA(S)\cap\Sig_\alpha\subseteq\Cl_\alpha(S).
\]
The two inclusions give
\[
  \CA(S)\cap\Sig_\alpha=\Cl_\alpha(S)
\]
for every chart \(\alpha\) and every local premise set \(S\subseteq\Sig_\alpha\).  Thus \(\CA\) is a global conservative realization.
\end{proof}

\begin{corollary}[Finite obstruction algorithm]
\label{cor:finite-obstruction-algorithm}
There is a finite algorithm which, given a finite closure atlas \(\A\), computes \(\Obs(\A)\).  For each chart \(\alpha\in A\) and each local subset \(S\subseteq\Sig_\alpha\), compute \(\CA(S)\) by iterating \(\GA\), form
\[
  E_{\alpha,S}
  =
  \bigl(\CA(S)\cap\Sig_\alpha\bigr)\setminus\Cl_\alpha(S),
\]
and output all triples \((\alpha,S,\varphi)\) with \(\varphi\in E_{\alpha,S}\).  The resulting finite set is exactly \(\Obs(\A)\).  Consequently, \(\A\) has a global conservative realization if and only if the algorithm returns the empty set.
\end{corollary}

\begin{proof}
There are finitely many charts, and each \(\Sig_\alpha\) has finitely many subsets.  For fixed \(\alpha\) and \(S\subseteq\Sig_\alpha\), the sequence
\[
  S=\GA^0(S)\subseteq\GA^1(S)\subseteq\GA^2(S)\subseteq\cdots
\]
is increasing and contained in the finite set \(\Sig_{\A}\).  It therefore stabilizes after finitely many steps, and its stable value is \(\CA(S)\).  The formula for \(E_{\alpha,S}\) then lists exactly the chart-visible excess consequences at \((\alpha,S)\).  Taking all charts and all local premise sets therefore outputs precisely \(\Obs(\A)\).  The final equivalence follows from Theorem~\ref{thm:complete-closure-atlas-obstruction}.
\end{proof}

\begin{remark}[Pseudocode form]
The finite procedure can be written as follows.
\begin{quote}
\begin{verbatim}
O := empty set
Sigma := union_alpha Sigma_alpha

for each alpha in A:
    for each S subset Sigma_alpha:
        X := S
        repeat:
            Y := X union union_beta cl_beta(X intersection Sigma_beta)
            if Y = X:
                break
            X := Y
        C := X

        E := (C intersection Sigma_alpha) minus cl_alpha(S)
        for each phi in E:
            O := O union {(alpha,S,phi)}

return O
\end{verbatim}
\end{quote}
The fixed point \(X\) reached by the loop is \(\CA(S)\).
\end{remark}

\begin{remark}[Finite complexity]
Let \(m=|\Sig_{\A}|\) and \(m_\alpha=|\Sig_\alpha|\).  The algorithm tests
\[
  \sum_{\alpha\in A}2^{m_\alpha}
\]
local subsets.  For each tested subset \(S\), the iteration of \(\GA\) stabilizes after at most \(m-|S|\leq m\) strict growth stages.  Thus the number of global iteration rounds is bounded by
\[
  m\sum_{\alpha\in A}2^{m_\alpha}.
\]
This is a finite decidability bound, not an efficiency claim.
\end{remark}

\begin{example}[A chart-visible obstruction from pairwise compatible overlaps]
\label{ex:pairwise-compatible-atlas-not-realizable}
This example adds a third chart to the propagation pattern above.  The added chart makes the propagated conclusion visible locally while refusing it locally.

Let the global sentence set be \(\{p,q,r\}\), and consider three charts
\[
  \Sig_{pq}=\{p,q\},
  \qquad
  \Sig_{qr}=\{q,r\},
  \qquad
  \Sig_{pr}=\{p,r\}.
\]
Let \(\Cl_{pq}\) be generated by \(p\vdash_{pq}q\), let \(\Cl_{qr}\) be generated by \(q\vdash_{qr}r\), and let \(\Cl_{pr}\) be identity closure on \(\{p,r\}\).  Thus the \(pr\)-chart explicitly has
\[
  p\nvdash_{pr} r.
\]
The pairwise overlaps are singleton languages:
\[
  \Sig_{pq}\cap\Sig_{qr}=\{q\},
  \qquad
  \Sig_{qr}\cap\Sig_{pr}=\{r\},
  \qquad
  \Sig_{pq}\cap\Sig_{pr}=\{p\}.
\]
On these singleton overlaps there is no direct pairwise disagreement: restricting either adjacent chart to the overlap gives the identity closure on the singleton language.  Nevertheless, the atlas has no global conservative realization.

Indeed, compute the atlas-generated closure of \(S=\{p\}\subseteq\Sig_{pr}\).  The \(pq\)-chart first adds \(q\):
\[
  \GA(\{p\})=\{p,q\}.
\]
Then the \(qr\)-chart adds \(r\):
\[
  \GA^2(\{p\})=\{p,q,r\}.
\]
Thus
\[
  \CA(\{p\})=\{p,q,r\}.
\]
Restricting this global closure back to the \(pr\)-chart gives
\[
  \CA(\{p\})\cap\Sig_{pr}=\{p,r\},
\]
whereas locally
\[
  \Cl_{pr}(\{p\})=\{p\}.
\]
Therefore
\[
  r\in\CA(\{p\})
  \qquad\text{but}\qquad
  r\notin\Cl_{pr}(\{p\}).
\]
So
\[
  (pr,\{p\},r)\in\Obs(\A).
\]
By Theorem~\ref{thm:complete-closure-atlas-obstruction}, no global conservative realization exists.  This example shows why pairwise overlap agreement is weaker than the obstruction criterion: the failure appears only after propagation through more than one chart.

\begin{figure}[!htbp]
\centering
\begin{tikzpicture}[
  sent/.style={circle, draw, minimum size=8mm},
  rule/.style={-{Latex[length=2mm]}, thick},
  denied/.style={-{Latex[length=2mm]}, thick, dashed},
  chart/.style={draw, rounded corners, inner sep=7pt},
  every node/.style={font=\small}
]
\node[sent] (p) at (0,0) {$p$};
\node[sent] (q) at (2.6,1.65) {$q$};
\node[sent] (r) at (5.2,0) {$r$};

\draw[rule] (p) -- node[above left] {$pq$} (q);
\draw[rule] (q) -- node[above right] {$qr$} (r);
\draw[denied] (p) -- node[below=1pt] {$p\nvdash_{pr} r$} (r);

\node[chart, fit=(p)(q), label={[font=\footnotesize]above left:$\Sigma_{pq}$}] {};
\node[chart, fit=(q)(r), label={[font=\footnotesize]above right:$\Sigma_{qr}$}] {};
\node[chart, fit=(p)(r), label={[font=\footnotesize]below:$\Sigma_{pr}$}] {};

\end{tikzpicture}
\caption{A chart-visible obstruction: propagation forces \(p\vdash r\), but the \(pr\)-chart does not validate it.}
\label{fig:chart-visible-obstruction}
\end{figure}
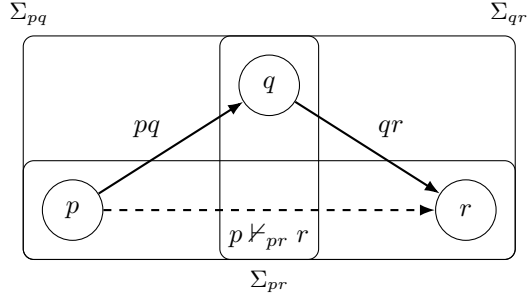

\end{example}

\begin{remark}[Pairwise compatibility is insufficient]
In the preceding example, no pair of charts directly disagrees on its overlap.  The failure is caused by the interaction
\[
  p\vdash_{pq}q,
  \qquad
  q\vdash_{qr}r.
\]
Any global closure extending the first two charts must satisfy \(p\vdash r\).  But the third chart, whose language contains \(p\) and \(r\), requires \(p\nvdash_{pr}r\).  The obstruction is therefore genuinely local-to-global: it appears only after consequences are propagated through the atlas.
\end{remark}

\begin{remark}[Premise sets in obstructions]
The obstruction set permits arbitrary finite local premise sets \(S\subseteq\Sig_\alpha\).  This is intentional.  Multi-premise local consequences can produce chart-visible failures not detected by singleton-premise paths.  Theorem~\ref{thm:complete-closure-atlas-obstruction} gives the closure-level realization problem for finite closure atlases; the next section treats the separate problem of gluing selected local closed theories.
\end{remark}

\section{Compatible Local Closed Theories}
\label{sec:closed-theory-gluing}

The obstruction theorem concerns closure systems: it asks when the local consequence relations of an atlas can be recovered as conservative restrictions of a single global closure relation.  We now turn to a related but simpler gluing problem for selected local closed theories.  Here the canonical closure \(\CA\) is fixed in advance, and the question is whether local closed theories can be represented by a single global closed set.  Exact agreement on overlaps, together with local closedness, is enough.

\begin{definition}[Compatible local closed theories and canonical union]
Let
\[
  \A=\{(\Sig_\alpha,\Cl_\alpha)\}_{\alpha\in A}
\]
be a finite closure atlas.  A \emph{local closed-theory family} on \(\A\) is a tuple
\[
  \mathbf T=(T_\alpha)_{\alpha\in A}
\]
such that \(T_\alpha\subseteq\Sig_\alpha\) and \(\Cl_\alpha(T_\alpha)=T_\alpha\) for every \(\alpha\in A\).

The family is \emph{overlap-compatible} if, for every pair \(\alpha,\beta\in A\),
\[
  T_\alpha\cap\Sig_{\alpha\beta}
  =
  T_\beta\cap\Sig_{\alpha\beta},
\]
where \(\Sig_{\alpha\beta}=\Sig_\alpha\cap\Sig_\beta\).  Thus every sentence visible in two charts has the same membership status in the corresponding local theories.

For such a family, its \emph{canonical union} is
\[
  \Tcup=\bigcup_{\alpha\in A}T_\alpha\subseteq\Sig_{\A}.
\]
If a global set \(T\subseteq\Sig_{\A}\) restricts exactly to every \(T_\alpha\), then necessarily \(T=\Tcup\).  Hence \(\Tcup\) is the only possible underlying set of an exact global representative of the local family.
\end{definition}

\begin{counterexample}[Overlap incompatibility prevents exact gluing]
\label{cex:overlap-incompatibility-prevents-gluing}
Let
\[
  \Sig_1=\{p,q\},
  \qquad
  \Sig_2=\{q,r\},
\]
with identity closures on both charts.  Take local closed theories
\[
  T_1=\{p,q\},
  \qquad
  T_2=\{r\}.
\]
Both are locally closed, but they disagree on the overlap \(\Sig_{12}=\{q\}\):
\[
  T_1\cap\Sig_{12}=\{q\},
  \qquad
  T_2\cap\Sig_{12}=\varnothing.
\]
No global set \(T\subseteq\{p,q,r\}\) can restrict exactly to both local theories.  If \(T\cap\Sig_1=T_1\), then \(q\in T\); if \(T\cap\Sig_2=T_2\), then \(q\notin T\).  Thus overlap compatibility is not a decorative hypothesis: it is necessary even before any closure operation is applied.
\end{counterexample}

\begin{lemma}[Union-restriction lemma]
\label{lem:union-restriction}
If \(\mathbf T=(T_\alpha)_{\alpha\in A}\) is overlap-compatible, then for every \(\alpha\in A\),
\[
  \Tcup\cap\Sig_\alpha=T_\alpha.
\]
\end{lemma}

\begin{proof}
The inclusion \(T_\alpha\subseteq\Tcup\cap\Sig_\alpha\) is immediate.  Conversely, let \(\varphi\in\Tcup\cap\Sig_\alpha\).  Since \(\varphi\in\Tcup\), there is some \(\beta\in A\) such that \(\varphi\in T_\beta\).  Because \(T_\beta\subseteq\Sig_\beta\) and \(\varphi\in\Sig_\alpha\), we have \(\varphi\in\Sig_{\alpha\beta}\).  By overlap compatibility,
\[
  T_\alpha\cap\Sig_{\alpha\beta}
  =
  T_\beta\cap\Sig_{\alpha\beta}.
\]
Thus \(\varphi\in T_\alpha\).  Hence \(\Tcup\cap\Sig_\alpha\subseteq T_\alpha\), proving the equality.
\end{proof}

\begin{theorem}[Canonical gluing theorem]
\label{thm:canonical-closed-theory-gluing}
Let \(\A\) be a finite closure atlas, and let \(\CA\) be its atlas-generated closure.  If \(\mathbf T=(T_\alpha)_{\alpha\in A}\) is an overlap-compatible local closed-theory family, then
\[
  \Tcup\cap\Sig_\alpha=T_\alpha
\]
for every \(\alpha\in A\), and
\[
  \CA(\Tcup)=\Tcup.
\]
Thus compatible local closed theories glue exactly under the canonical atlas-generated closure.
\end{theorem}

\begin{proof}
The restriction identity is Lemma~\ref{lem:union-restriction}.  It remains to prove that \(\Tcup\) is \(\CA\)-closed.  By the same lemma,
\[
  \Tcup\cap\Sig_\alpha=T_\alpha
\]
for every \(\alpha\).  Since each local theory is closed,
\[
  \Cl_\alpha(\Tcup\cap\Sig_\alpha)
  =
  \Cl_\alpha(T_\alpha)
  =
  T_\alpha.
\]
Therefore the atlas step operator satisfies
\[
  \GA(\Tcup)
  =
  \Tcup\cup\bigcup_{\alpha\in A}\Cl_\alpha(\Tcup\cap\Sig_\alpha)
  =
  \Tcup\cup\bigcup_{\alpha\in A}T_\alpha
  =
  \Tcup.
\]
Thus \(\Tcup\) is fixed by \(\GA\).  All subsequent iterates are again \(\Tcup\), so the stable value defining \(\CA(\Tcup)\) is \(\Tcup\) itself.
\end{proof}

\begin{remark}[Relation to the obstruction theorem]
Theorem~\ref{thm:canonical-closed-theory-gluing} does not require \(\Obs(\A)=\varnothing\).  It is a statement about the closure operator \(\CA\), which exists for every finite closure atlas.  If the obstruction set is empty, then Theorem~\ref{thm:complete-closure-atlas-obstruction} says that \(\CA\) is also a global conservative realization of the atlas.  In that case the compatible local family glues into a closed theory of the canonical conservative realization.  If the obstruction set is nonempty, the same set \(\Tcup\) is still closed under \(\CA\), but \(\CA\) is not conservative over all charts.
\end{remark}

\begin{example}[Successful canonical gluing]
\label{ex:successful-canonical-gluing}
Let
\[
  \Sig_1=\{p,q\},
  \qquad
  \Sig_2=\{q,r\}.
\]
Let \(\Cl_1\) be generated by \(p\vdash_1 q\), and let \(\Cl_2\) be generated by \(q\vdash_2 r\).  Choose local closed theories
\[
  T_1=\{p,q\},
  \qquad
  T_2=\{q,r\}.
\]
They are closed because \(T_1\) already contains the \(\Cl_1\)-consequence of \(p\), and \(T_2\) already contains the \(\Cl_2\)-consequence of \(q\).  The overlap language is
\[
  \Sig_{12}=\{q\},
\]
and both local theories restrict to \(\{q\}\) on that overlap.  Hence the family is overlap-compatible.

The canonical union is
\[
  \Tcup=T_1\cup T_2=\{p,q,r\}.
\]
By the union-restriction lemma,
\[
  \Tcup\cap\Sig_1=T_1,
  \qquad
  \Tcup\cap\Sig_2=T_2.
\]
Moreover, each chart sees a closed local theory when it restricts \(\Tcup\) to its own language, so one atlas step adds nothing:
\[
  \GA(\Tcup)=\Tcup.
\]
Therefore
\[
  \CA(\Tcup)=\Tcup.
\]
The example displays the exact content of the gluing theorem: compatibility gives the prescribed restrictions, and local closedness makes the canonical union globally closed under \(\CA\).
\end{example}

\begin{remark}[No second canonical obstruction]
Theorem~\ref{thm:canonical-closed-theory-gluing} is elementary.  Its role is to separate two questions.  The nontrivial obstruction problem is the closure-level problem of Section~\ref{sec:obstructions}: whether local consequence relations globalize conservatively.  Once the canonical closure \(\CA\) has been fixed, exactly overlap-compatible closed local theories have no further canonical gluing obstruction.  Stronger questions about gluing under larger noncanonical conservative realizations lead to the conservative-completion problems mentioned in the outlook, but they are not part of the present paper.
\end{remark}

\section{Discussion, Limitations, and Outlook}
\label{sec:conclusion}

The preceding sections establish a finite closure-theoretic framework for comparing local and global closure data.  The main object is a finite closure atlas
\[
  \A=\{(\Sig_\alpha,\Cl_\alpha)\}_{\alpha\in A},
\]
whose chart closures propagate through overlaps by the step operator \(G_{\A}\).  Iterating this operator produces the atlas-generated closure \(\CA\), the least global closure operator extending the local chart closures.

The main theorem gives an exact finite answer to the conservative-globalization problem.  The obstruction set records precisely where this least global propagation creates a chart-visible consequence not already present in the corresponding chart.  The local-to-global obstruction criterion proves the equivalence
\[
\begin{aligned}
  \A\text{ has a global conservative realization}
  &\Longleftrightarrow
  \CA\text{ is a global conservative realization} \\
  &\Longleftrightarrow
  \Obs(\A)=\varnothing.
\end{aligned}
\]
Thus, in the finite setting considered here, there is no hidden obstruction beyond the chart-visible excess detected by \(\CA\).  The criterion is also effective: one computes \(\CA(S)\) for finite chart-local premise sets and compares the result with the relevant local closure.

The indexed truth-space layer records one interpretation of the same closure data.  An indexed truth space represents an element \(\varphi\) by its region
\[
  V_{\X}(\varphi)=\{i\in I:\varphi\in\tau(i)\}
\]
inside a selected finite family of closed theories.  The full truth-region representation theorem shows that, when every closed theory is used as an index, region inclusion recovers the original closure consequence relation exactly.  Reduced indexed spaces are therefore not neutral summaries: by omitting separating closed theories, they may validate region consequences that are not consequences of the underlying closure operator.

The decision-bookkeeping layer is also separate from the atlas obstruction theorem.  For an element \(a\) and a designated formal opposite \(\bar a\), the general decomposition is the four-region partition
\[
  I=P_a\sqcup N_a\sqcup B_a\sqcup U_a,
\]
where the both-region \(B_a=V(a)\cap V(\bar a)\) records indices whose assigned closed theories contain both designated elements.  This is only membership data in the chosen indexed space.  It is not a semantics of negation, and it does not by itself assert anything about contradiction, paraconsistency, or classical negation.

The gluing theorem is deliberately separate from the closure-level obstruction.  Once \(\CA\) has been fixed, overlap-compatible local closed theories glue by ordinary union:
\[
  \Tcup=\bigcup_{\alpha\in A}T_\alpha.
\]
Compatibility guarantees that \(\Tcup\) restricts back to the prescribed local theories, and local closedness guarantees
\[
  \CA(\Tcup)=\Tcup.
\]
Thus exact compatible local closed theories have no further canonical gluing obstruction.  The nontrivial obstruction studied in the paper lies in conservative globalization of local consequence relations, not in the final union step for already compatible closed local data.

\subsection*{Limitations and nonclaims}

The framework is intentionally finite.  Infinite universes or infinite atlases would require additional hypotheses, transfinite iteration, or a different construction of least global closure.  None of those extensions is treated here.  The paper also does not define a topology of truth regions, even though the terminology of regions suggests possible later geometric or topological refinements.  At this stage, truth regions are subsets of a finite index set, and closure atlases are finite combinatorial objects.

The notation \(\bar a\) denotes a designated formal opposite, not full logical negation unless extra structure is supplied.  Likewise, the both-region \(B_a\) should not be read as a claim about contradiction or about which logic of contradiction is appropriate.  The paper deliberately avoids using paradoxes or independence phenomena as arguments against ordinary classical reasoning.  Nothing proved here implies that classical logic is false, that the law of excluded middle fails globally, or that standard foundations are invalid.

The results should also not be read as a general semantics of mathematical truth.  The paper studies closure operators and finite families of closed theories as mathematical data.  Its claims are structural: how region representations behave, how conservative extensions preserve old consequence, how closure atlases generate global closure, and when that global closure is conservative over each chart.

\subsection*{Relation to the broader program}

The phrase \emph{Liberal Mathematics} is best understood here as the name of a broader research program concerned with indexed, local, and theory-relative mathematical structure.  The present article proves only one finite closure-theoretic component of that program.  Its content is independent of any stronger philosophical thesis: the definitions and theorems stand or fall as claims about finite closure systems, finite indexed truth spaces, and finite closure atlases.

This restraint is important.  The paper does not propose a replacement for classical logic or for standard foundations.  It gives a small mathematical language for a specific closure-theoretic phenomenon: elements may be represented by regions across selected closed theories, and local closure systems may or may not admit conservative global realization.  That is the level at which the contribution should be evaluated.

\subsection*{Outlook: conservative completions}

The canonical gluing theorem uses the least global closure \(\CA\).  A natural next problem concerns noncanonical global conservative realizations.  Such realizations may add mixed-premise consequences that are invisible inside every single chart while still preserving all chart-internal consequence relations.

For example, take identity charts on
\[
  \Sig_1=\{p,r\},
  \qquad
  \Sig_2=\{q,r\},
\]
and local closed theories \(T_1=\{p\}\) and \(T_2=\{q\}\).  The canonical union is \(\Tcup=\{p,q\}\), and \(r\) is not forced by the atlas-generated closure.  Nevertheless, a noncanonical conservative global closure may add the mixed rule \(\{p,q\}\Rightarrow r\), because no single chart sees both premises at once.

This points toward a later theory of conservative global completions of compatible local data.  Such a theory would ask which mixed consequences can be added without changing any chart-internal closure relation, and how the resulting finite family of admissible global closures should be organized.  Those questions are not used in the present paper.  The present contribution is the first finite layer: indexed truth spaces, four-region decision decompositions, conservative extension, closure-atlas obstruction, and canonical gluing under the atlas-generated closure.

\bibliographystyle{amsplain}
\bibliography{refs}
\end{document}